\documentclass[runningheads]{llncs}
\usepackage[T1]{fontenc}
\usepackage{graphicx}
\usepackage{amsfonts}
\usepackage{amsmath}
\usepackage[hypertexnames=false]{hyperref}
\usepackage{braket}
\usepackage{multirow}
\usepackage{diagbox}
\usepackage{pifont}
\usepackage{algorithm}
\usepackage{algpseudocode}
\usepackage{makecell}

\usepackage{enumitem}
\usepackage{float}
\usepackage{tikz}
\usepackage{pgfplots}
\pgfplotsset{compat=1.18}
\usepackage{xcolor}
\usepackage{listings}
\usepackage{subcaption}
\usepackage{orcidlink}
\usepackage{marvosym}

\lstdefinelanguage{MyPython}{
  language=Python,
  morekeywords={new_qubit,measure}, 
}

\begin{document}
\title{QSeqSim: A Symbolic Simulator for Qiskit While Loops Using Sequential Quantum Circuits}

\titlerunning{QSeqSim}


%
%

\author{Zihao Li\inst{1,2}~\orcidlink{0009-0002-1768-6662}, Ji Guan\inst{1}\textsuperscript{(\Letter)}~\orcidlink{0000-0002-3490-0029}, Mingsheng Ying\inst{3}~\orcidlink{0000-0003-4847-702X}
}

\authorrunning{Z. Li et al.}

\institute{
Key Laboratory of System Software (Chinese Academy of Sciences), Institute of Software, Chinese Academy of Sciences, Beijing, 100190, China\\
\email{lizh@ios.ac.cn, guanj@ios.ac.cn}
\and
University of Chinese Academy of Sciences, Beijing, 101408, China
\and
Centre for Quantum Software and Information, University of Technology Sydney, Ultimo NSW 2007, Australia\\
\email{Mingsheng.Ying@uts.edu.au}
}

%
\maketitle              
\begin{abstract}
We present a tool \emph{QSeqSim}, a Qiskit-integrated symbolic backend that fills the current gap of having no Qiskit-native support for simulating \texttt{while}-loop quantum programs and their induced sequential quantum circuits. \emph{QSeqSim} takes Qiskit \texttt{QuantumCircuit} objects, translates them into OpenQASM~3 code, and organises the resulting program into a combination of combinational, dynamic, and sequential circuits, thereby assigning \texttt{while}-loops a precise sequential circuit semantics with explicit internal and external qubits. Building on this semantics, \emph{QSeqSim} adopts a Binary Decision Diagram (BDD)-based symbolic representation and integrates weighted model counting to compute measurement probabilities efficiently by exploiting sharing in structured and sparse BDDs. On top of this Boolean backbone, it introduces dedicated symbolic operators for quantum state composition and state retention, thereby enabling efficient symbolic execution of sequential quantum circuits. Our experiments demonstrate that \emph{QSeqSim} scales to substantial \texttt{while}-induced sequential circuits; in particular, in the quantum random walk benchmark we successfully simulate circuits with over 1000 qubits for more than 10 loop iterations.

\emph{QSeqSim} is available at \url{https://github.com/Veri-Q/QSeqSim}.


\end{abstract}
\section{Introduction}

The formal methods community has devoted substantial effort to the analysis and
verification of quantum circuits and programs, both at the foundational level
and in the form of automated tools.
There is now a considerable body of work on quantum Hoare logic~\cite{ying2012floyd,zhou2019applied,liu2019formal,unruh2019quantum},
as well as on other quantum program logics,  deductive
verification and semantics for languages with
measurements and classical control~\cite{selinger2004towards,ying2024foundations}.
On the circuit side, equivalence checking and optimisation techniques based on
decision diagrams~\cite{zulehner2018advanced,burgholzer2020advanced,tsai2021bit,wei2022accurate,chen2025sliqsim,wang2025feynmandd},
ZX-calculus~\cite{kissinger2019pyzx,peham2022equivalence},
tree automata~\cite{chen2023autoq,chen2023automata,chen2025autoq},
and model counting~\cite{mei2024simulating,mei2024equivalence,quist2024advancing}
have been developed, and model checking and
abstract-interpretation approaches for quantum and probabilistic systems have
also been investigated~\cite{gay2005probabilistic,papanikolaou2009model,feng2013model,ying2021model}.

In parallel with these developments, Qiskit~\cite{javadi2024quantum}, one of the most
widely used quantum programming frameworks, has recently enriched the
traditional circuit model with classical control, allowing users to express
quantum algorithms with high-level constructs such as conditional branches
(\texttt{if}/\texttt{else}, \texttt{switch}) and loops (\texttt{for},
\texttt{while}). In particular, the availability of \texttt{while}-loops makes
it possible to encode measurement-controlled iteration patterns such as
repeat-until-success (RUS) schemes~\cite{paetznick2013repeat,bocharov2015efficient},
weakly measured variants of Grover's search~\cite{andres2022weakly}, and quantum
random walks with feedback~\cite{kendon2003decoherence}. At the level of
abstract syntax, this brings mainstream quantum programming closer to the
control-flow constructs long studied in quantum programming
languages~\cite{selinger2004towards} and quantum Hoare logics~\cite{ying2012floyd}.

On the execution side, however, the situation is markedly less mature: 

{\emph{At present, there is no Qiskit-native tool that can simulate Qiskit}
\texttt{while}\emph{-loop programs.}}

Qiskit-Aer, the default simulator in the Qiskit ecosystem, currently
\emph{does not} support the \texttt{while}-loop construct at all.
In particular, attempting to run the official Qiskit API examples that use
\texttt{QuantumCircuit.while\_loop} through Aer results in parsing or
compilation errors, indicating that such control flow cannot be lowered to a
circuit that Aer is able to simulate. A similar limitation holds for the Qiskit
cloud backends: when submitting circuits that contain \texttt{while}-loop
constructs, the service explicitly reports that the available real-device
backends \emph{do not} support \texttt{while}-loop constructs.
In other words, although Qiskit and OpenQASM~3~\cite{cross2022openqasm3} expose
\texttt{while}-loops at the programming and IR levels, neither the standard
Qiskit simulator nor the current hardware backends can execute these constructs,
and there is, at present, no Qiskit-native execution engine that realises the
intended iterative semantics of non-trivial \texttt{while}-loop programs.

To bridge this gap, we present \emph{QSeqSim}, a symbolic simulation backend for Qiskit programs with classical control, with a particular emphasis on \texttt{while}-loop–induced sequential behaviour. On the semantic side, we interpret Qiskit \texttt{while}-loops as \emph{sequential quantum circuits (SQCs)} with an explicit partition into external and internal qubits, thereby making precise the notion of quantum state feedback across loop iterations. On the algorithmic side, we build on a BDD-based representation of quantum circuits~\cite{tsai2021bit} and extend it from purely combinational circuits to programs with structured control flow: \texttt{if}/\texttt{switch} constructs are treated as dynamic circuits with mid-circuit measurements, \texttt{for} loops with static bounds as repeated combinational circuits, and \texttt{while} loops as sequential circuits. At the Boolean level, we introduce new symbolic operators for \emph{state composition} and \emph{state retention}, which model how external inputs are combined with the current internal state and how the post-measurement internal state is fed back into the next iteration. These operators are integrated with existing Boolean gate update rules and a Boolean mid-circuit measurement operator into a unified BDD-based engine that executes mixed combinational, dynamic, and sequential fragments according to a small-step operational semantics, while systematically exploiting weighted model counting~\cite{chavira2008probabilistic} techniques to carry out probability computation efficiently on structured and sparse BDDs.

{\textbf{Contributions.}}
This paper makes the following main contributions.

\begin{enumerate}
  \item \emph{First Qiskit backend with direct} \texttt{while}\emph{-loop support.}
        We present \emph{QSeqSim}, to the best of our knowledge, the first backend
        that directly executes Qiskit programs containing \texttt{while}-loops,
        by lowering them to OpenQASM~3 and giving an executable small-step
        semantics for the resulting control flow.

  \item \emph{From combinational to sequential Boolean models.}
        \emph{QSeqSim} extends a BDD-based encoding of combinational quantum circuits~\cite{tsai2021bit} with symbolic operators for state composition and retention, lifting the Boolean model from purely combinational to sequential circuits so as to support \texttt{while}-loops.

   \item \emph{Efficient and scalable symbolic simulation of} \texttt{while}\emph{-loop programs (SQCs).}
      Leveraging mature BDD backends and model-counting techniques for efficient probability evaluation,
      \emph{QSeqSim} symbolically simulates quantum circuits induced by \texttt{while}-loops, scaling beyond
      existing tools~\cite{chen2025autoq,wang2021equivalence}.
\end{enumerate}

\subsection{Related Work}

\textbf{\emph{Quantum Programming Backends.}} A rich ecosystem of quantum programming frameworks and simulators has emerged in recent years, including general-purpose stacks such as Qiskit~\cite{javadi2024quantum}, Cirq~\cite{cirq}, ProjectQ~\cite{steiger2018projectq}, CUDA-Q~\cite{cudaq}, and Amazon Braket~\cite{gonzalez2021cloud}, as well as specialised high‑performance simulators like QuEST~\cite{jones2019quest}, Qulacs~\cite{suzuki2021qulacs}, and Intel-QS (IQS)~\cite{guerreschi2020intel}. However, there remains a clear gap between the expressiveness of modern quantum programming models and the capabilities of current execution backends. On the language side, Qiskit has recently enriched its model with high-level classical control constructs such as \texttt{if\_else}, \texttt{for\_loop}, and in particular \texttt{while\_loop}, enabling natural encodings of measurement‑controlled iteration patterns. On the backend side, however, mainstream backends effectively lack native support for \texttt{while}-loops: Qiskit-Aer cannot currently execute \texttt{QuantumCircuit.while\_loop}, and Qiskit’s cloud hardware backends likewise report \texttt{while}-loops as unsupported. \emph{QSeqSim} is designed to fill precisely this gap. Starting from \texttt{QuantumCircuit} objects, it translates circuits to OpenQASM~3, assigns the resulting programs a sequential-circuit semantics, and uses a BDD-based engine to execute \texttt{while}-loop programs at a scale that makes it suitable as an efficient backend.


\paragraph{\textbf{Simulation Techniques for Quantum Circuits.}} There is also extensive work in the formal methods community on the simulation of quantum circuits~\cite{zulehner2018advanced,tsai2021bit,chen2025sliqsim,wang2025feynmandd,kissinger2019pyzx,chen2023autoq,chen2023automata,mei2024simulating}. 
However, these approaches remain essentially confined to combinational and very
limited dynamic settings; they do not extend their models to \emph{sequential
quantum circuits} with explicit state feedback, and thus cannot support the
Qiskit \texttt{while}-loop construct.
\emph{QSeqSim} builds directly on the Boolean framework of~\cite{tsai2021bit}, and
further integrates weighted model counting for probability queries, while also
extending the framework with new symbolic operators for state composition and
state retention. This yields a unified BDD-based engine that can simulate,
within a single Qiskit program, combinational, dynamic, and sequential
fragments and thereby natively support \texttt{while}-loop induced sequential
behaviour.

\section{Background}
\label{sec:background}
This section provides the necessary background on Qiskit, OpenQASM~3, and SQCs that underpins our semantics and implementation.

Figure~\ref{fig:background-pipeline} structures Qiskit \texttt{while}-loop programs into a three-stage pipeline that encompasses the programming language layer, intermediate representation layer, and circuit layer. It also maps each quantum stage to a well-known classical counterpart, allowing the quantum components to be understood through direct comparison with their classical equivalents.
\begin{figure}[h]
  \centering
  \includegraphics[width=\textwidth]{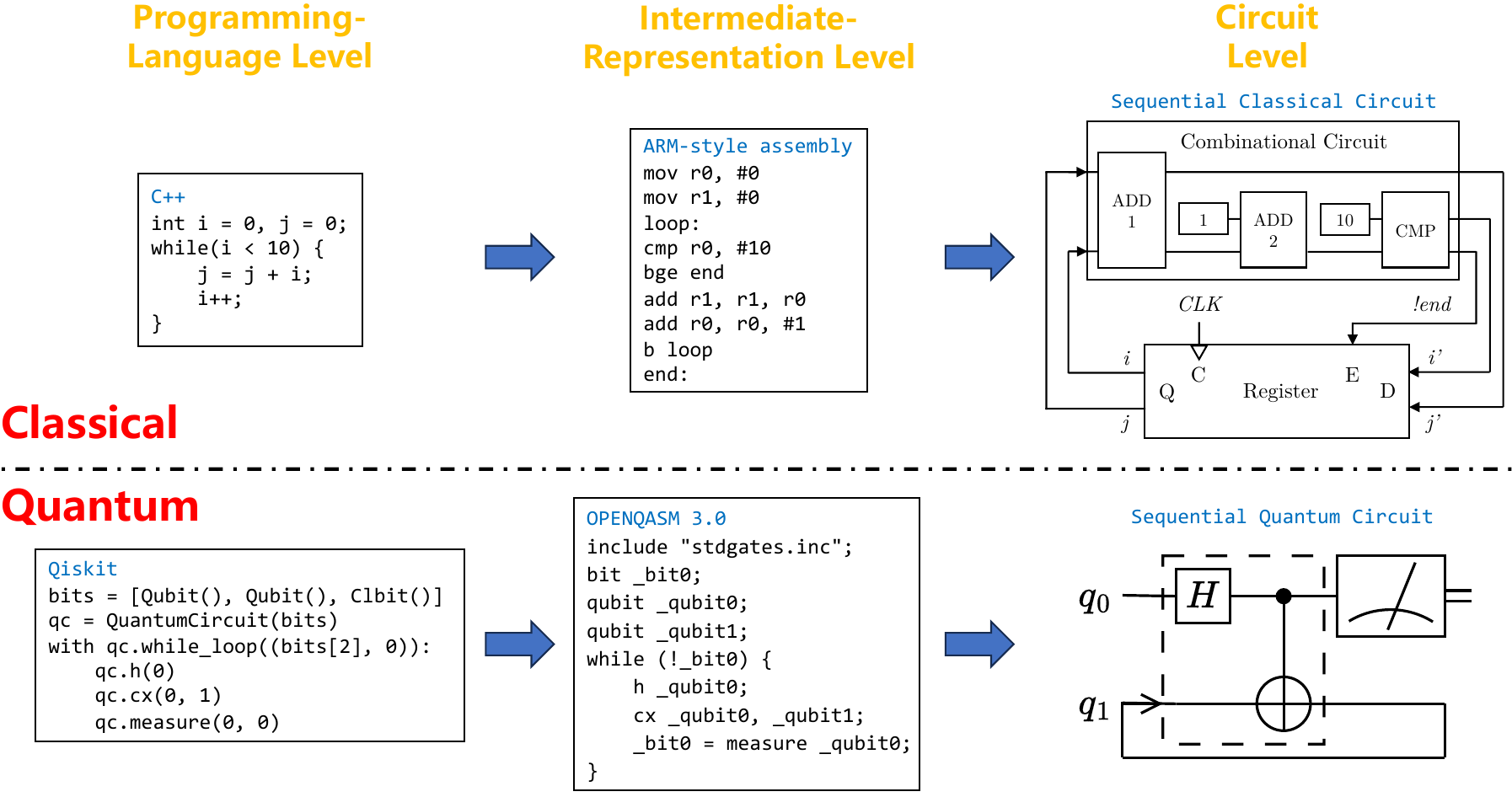}
  \caption{Classical–quantum analogy from programs to IRs to circuits.}
  \label{fig:background-pipeline}
\end{figure}

\paragraph{\textbf{Programming-Language Level.}}
In the top left of Figure~\ref{fig:background-pipeline}, a C++ \texttt{while}
loop over integers $(i,j)$ is the programmer-facing description of an iterative
computation.
In the bottom left, the Qiskit code shows a concrete quantum analogue: a RUS implementation of the $X$ gate.
Here, the \texttt{QuantumCircuit.while\_loop} construct repeatedly invokes a
small quantum subroutine and measures qubit $q_0$ to obtain a classical outcome;
the loop terminates only when this measurement result is $1$, at which point
the target qubit $q_1$ has effectively been subjected to an $X$ operation.
This example illustrates how Qiskit \texttt{while}-loops are used in practice to
express measurement-controlled iteration patterns.

\paragraph{\textbf{Intermediate-Representation (IR) Level.}}
The C++ \texttt{while} loop is then compiled to an ARM-style assembly fragment
(top middle of Figure~\ref{fig:background-pipeline}).
This assembly is a control-flow oriented IR that is easier to analyse and
further compile than the original C++. Similarly, Qiskit compiles the Python-level program to an \emph{internal} IR
represented by a \texttt{QuantumCircuit} object.
This object records the gates, measurements, and classical control structure in
a uniform form.
From there, Qiskit can export an \emph{external} IR in the form of OpenQASM~3
code~\cite{openqasm_live_spec,cross2022openqasm3} (bottom middle of Figure~\ref{fig:background-pipeline}): a quantum assembly language with a formally defined syntax
and semantics. 
In this work we treat OpenQASM~3, in the restricted fragment generated by Qiskit, as the semantic interface for \texttt{QuantumCircuit}.

\paragraph{\textbf{Circuit Level.}}
On the classical side (top right of Figure~\ref{fig:background-pipeline}), the
ARM-style assembly fragment is implemented as a \emph{sequential classical
circuit}: a combinational logic block composed of standard components such as
adders (\textsc{ADD}) and a comparator (\textsc{CMP}),
together with a register and feedback wires. On the quantum side (bottom right of Figure~\ref{fig:background-pipeline}), we
obtain the analogous notion of a \emph{sequential quantum circuit}.
The dashed box is the loop body: in this example it is a fixed subcircuit
consisting of a Hadamard on $q_0$ followed by a controlled-$X$ from $q_0$ to
$q_1$, and one pass through this box corresponds to a single loop iteration.
After the body, $q_0$ is measured and the outcome is used as the loop guard:
a result of $0$ triggers another iteration, whereas a result of $1$ terminates
the loop, exactly mirroring the role of the comparator in the sequential
classical circuit.
The wire for $q_1$ passes through the body and feeds back to its input,
indicating that the quantum state on $q_1$ is preserved and updated across
iterations, in direct analogy with the classical register that stores $i$ and
$j$.
\section{Implementation Challenges}
\label{sec:implementation-challenges}
Filling the current gap of having no Qiskit-native tool that can simulate
\texttt{while}-loop programs raises three main specific challenges.
\paragraph{\textbf{Challenge 1}: Matching Qiskit} \texttt{while}\emph{-loops with sequential semantics.}
From the circuit point of view, a \texttt{while}-loop induces sequential
behaviour: quantum state is fed back from one iteration to the next, and the
guard is evaluated from intermediate measurement outcomes.
The sequential quantum circuit (SQC) model of Wang et al.~\cite{wang2021equivalence}
captures this by iterating a fixed block over \emph{external} and \emph{internal}
qubits, with only the external ones measured each round.
Realistic Qiskit \texttt{while}-loops, however, both exceed and restrict this
abstraction: their loop bodies may contain mid-circuit measurements and nested
\texttt{if}/\texttt{else} branches not covered by the original SQC definition,
yet they arise from concrete \texttt{QuantumCircuit} programs and do not exploit
arbitrary per-iteration external inputs.

\noindent\emph{\textbf{Our approach.}}
We adopt a Qiskit-oriented generalisation of SQCs with an explicit split between
internal and external qubits and treat each \texttt{while}-loop body as a fixed
dynamic subcircuit with mid-circuit measurements and classical control, executed
unchanged in every iteration.
This generalised SQC view is realised concretely by the CQC/DQC/SQC
decomposition of the \emph{Parser} and the small-step operational semantics implemented
by the \emph{Simulator} (Section~\ref{sec:tool-architecture}, the \emph{Parser} and
Algorithm~\ref{alg:bddsim}).
\paragraph{\textbf{Challenge 2}: Symbolic simulation of sequential quantum circuits.}
Even with a suitable semantics and IR, simulating \texttt{while}-loops with
internal measurements and branches remains difficult.
Sequential circuits involve feedback across iterations, entangled internal
state, and dynamic bodies with nested \texttt{if}/\texttt{else} depending on
mid-circuit outcomes.
Naive statevector simulation scales poorly in both qubits and iterations, and
existing decision-diagram-based techniques mostly target combinational or, at
best, non-feedback dynamic circuits; they do not extend their models to general
sequential behaviour with explicit state feedback.

\noindent\emph{\textbf{Our approach.}}
Building on the BDD-based Boolean encoding of \cite{tsai2021bit}, we extend the
symbolic model from purely combinational circuits to general SQCs by introducing
operators for state composition and state retention, and by supporting external
inputs and internal state feedback across iterations.
The resulting SQC engine is implemented in the \emph{Kernel} component
\texttt{BDDSeqSim}, whose structure is detailed in
Section~\ref{sec:tool-architecture} (Kernel and Algorithm~\ref{alg:seqsim}).
\paragraph{\textbf{Challenge 3}: Measurement probabilities at scale.}
{\sloppy
Symbolic simulation of \texttt{while}-loop programs requires repeated evaluation
of measurement probabilities, especially for mid-circuit measurements and for
loop guards, often over large yet structured decision diagrams.
In existing BDD-based simulation techniques~\cite{tsai2021bit}, such
probabilities are obtained via a hyper-function construction over the decision
diagrams; however, for circuits with large numbers of qubits and deep structure
this approach becomes a major bottleneck and can dominate the overall runtime,
even when the underlying BDDs remain compact.
This raises a specific scalability challenge: how to obtain exact measurement
probabilities while preserving and exploiting the sharing and sparsity
structure of the BDDs generated by large quantum circuits.
\par}

\noindent\emph{\textbf{Our approach.}}
We incorporate model counting techniques~\cite{chavira2008probabilistic,mei2024simulating,mei2024equivalence} from formal methods and adopt weighted
model counting on structured and sparse BDDs as the core mechanism for
evaluating measurement probabilities. This counting-based engine implements the
\textsc{GetProb} interface used by both the \emph{Simulator} and the SQC \emph{Kernel}, as
detailed in Section~\ref{sec:tool-architecture} (Kernel, weighted model counting for
probability evaluation).

\section{Tool Architecture}
\label{sec:tool-architecture}

This section describes the architecture of \emph{QSeqSim} and explains how its
components jointly address the challenges identified in
Section~\ref{sec:implementation-challenges}.

Figure~\ref{fig:arch} presents the overall architecture of \emph{QSeqSim}.
Starting from a Qiskit program written in Python, \emph{QSeqSim} operates over
Qiskit’s internal intermediate representation (IR), namely the
\texttt{QuantumCircuit} object, and processes it through three core modules:
a \emph{Parser}, a \emph{Simulator}, and a Boolean \emph{Kernel}.
At a high level, the \emph{Parser} realises the compilation phase, whereas the
Simulator and Kernel jointly implement the symbolic execution back-end.
Together, the \emph{Parser} and Simulator instantiate the generalised SQC semantics
for Qiskit \texttt{while}-loops (Challenge~1), while the \emph{Kernel} provides the
symbolic machinery for SQCs (Challenge~2) and scalable probability computation
(Challenge~3).

The end-to-end workflow, which corresponds to Figure~\ref{fig:arch}, is as
follows.

\begin{figure}[h]
  \centering
  \includegraphics[width=\textwidth]{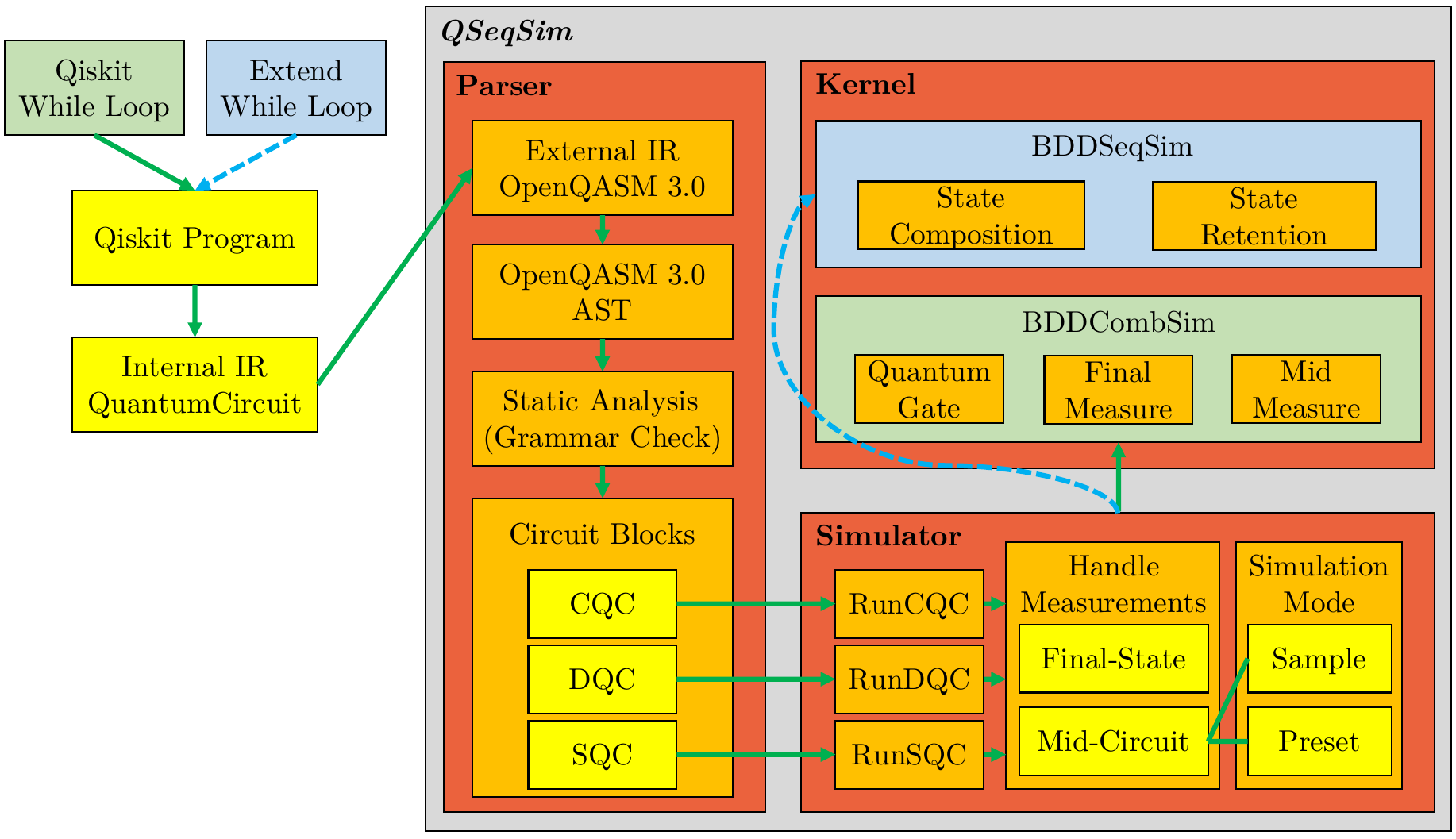}
  \caption{\emph{QSeqSim} workflow.}
  \label{fig:arch}
\end{figure}

\paragraph{\textbf{From Qiskit program to internal IR.}}
The user specifies a circuit via the Qiskit Python API.
Qiskit lowers it to a \texttt{QuantumCircuit} object, which is both its
internal IR and the entry point to \emph{QSeqSim}.

\paragraph{\textbf{Parser (compile phase).}}
Given a \texttt{QuantumCircuit}, the \emph{Parser} performs the following steps:
\begin{enumerate}[leftmargin=*]
  \item \emph{OpenQASM~3 export.}
        We invoke Qiskit’s \texttt{qasm3.dumps} routine to obtain an \emph{external}
        IR in the form of an OpenQASM~3 program.

  \item \emph{AST construction.}
        The resulting OpenQASM~3 code is parsed by the official OpenQASM~3
        parser into an abstract syntax tree (AST).

  \item \emph{Static analysis.}
        On this AST, \emph{QSeqSim} performs the following analyses:
        (i) it checks that all operations belong to the OpenQASM~3 fragment
        supported by the tool;
        (ii) it verifies that every unitary can be decomposed into the
        Clifford+T gate set implemented by the \emph{Kernel};
        (iii) it classifies each measurement as \emph{mid-circuit} or
        \emph{final}, depending on whether its outcome is used in any
        subsequent control condition; and
        (iv) it partitions the program into an ordered sequence of circuit
        blocks, each labelled as CQC, DQC or SQC. In particular, although we present the result as an ordered block sequence for
simplicity, the \emph{Parser} preserves the underlying \emph{hierarchical} block
structure induced by nested \texttt{if}/\texttt{switch}/\texttt{while} constructs.
\end{enumerate}

The resulting block sequence preserves the original program order of the Qiskit circuit and constitutes the sole input to the symbolic backend. By explicitly identifying CQC, DQC, and in particular SQC blocks for loop bodies, the \emph{Parser} provides exactly the structural interface required for the sequential semantics of Qiskit \texttt{while}-loops (Challenge~1).

\paragraph{\textbf{Simulator (operational semantics).}}
The \emph{Simulator} provides an operational interpretation of the block sequence
produced by the \emph{Parser}.
CQC blocks are treated as straight-line gate segments;
DQC blocks implement branching controlled by measurement outcomes; and
SQC blocks capture \texttt{while}-loops, which are unfolded by repeatedly
executing their bodies until the loop guard, read from classical bits,
evaluates to false.
Algorithm~\ref{alg:bddsim} formalises this behaviour in high-level
pseudo-code and instantiates the small-step semantics for mixed
combinational, dynamic, and sequential fragments, completing the realisation of
Challenge~1 on top of the \emph{Parser}’s CQC/DQC/SQC decomposition.

\begin{algorithm}[htp]
  \caption{\textsc{Simulator}$(\mathcal{B}, n, r, m)$}
  \label{alg:bddsim}
  \textbf{Input:}
      Parsed block sequence $\mathcal{B} = [B_1,\dots,B_K]$;
      number of qubits $n$;
      BDD precision parameter $r$;
      measurement mode $\textit{m} \in \{\textit{sample},\textit{preset}\}$. \\
    \textbf{Output:}
      BDD-encoded final quantum state $s$;
      path probability $p_{\text{global}}$.
  \begin{algorithmic}[1]
    \State $s \gets \textsc{BDDCombSim}(n, r)$ \Comment{initialise BDD kernel}
    \State initialise $s$ to $\lvert 0 \cdots 0 \rangle$
    \State $p_{\text{global}} \gets 1$
    \State $\textit{clbits} \gets$ empty map \Comment{classical store}
    \State \textsc{ExecuteBlocks}$(\mathcal{B})$
    \State \Return $[s, p_{\text{global}}]$
    \Statex
    \Procedure{ExecuteBlocks}{$\mathcal{B}'$}
      \For{each block $B$ in $\mathcal{B}'$ in program order}
        \If{$B$ is a CQC block}
          \For{each operation $G$ in $B$ in program order}
            \If{$G$ is a measurement $q \rightarrow c$}
              \State \textsc{HandleMeasure}$(q, c, \textsc{IsFinalMeas}(G))$
            \Else
              \State apply the Boolean gate update for $G$ on $s$
            \EndIf
          \EndFor
        \ElsIf{$B$ is a DQC block}
          \State read guard bits from $\text{clbits}$ and compute switch value $v$
          \State select branch $\mathcal{B}_{\text{sub}}$ of $B$ according to $v$
          \State \textsc{ExecuteBlocks}$(\mathcal{B}_{\text{sub}})$
        \ElsIf{$B$ is an SQC block} \Comment{while loops}
          \While{loop condition of $B$ evaluated from $\text{clbits}$ holds}
            \State let $\mathcal{B}_{\text{body}}$ be the body-block sequence of $B$
            \State \textsc{ExecuteBlocks}$(\mathcal{B}_{\text{body}})$
          \EndWhile
        \EndIf
      \EndFor
    \EndProcedure
    \Statex
    \Procedure{HandleMeasure}{$q, c, \textit{isFinal}$}
      \If{$\textit{isFinal}$}
        \State record $(q \rightarrow c)$ as a deferred final measurement
      \Else
        \State $(p_0^{\text{joint}}, p_1^{\text{joint}}) \gets \textsc{GetProb}(s, q)$ \Comment{query prob.\ of $q$}
        \State $p_0 \gets p_0^{\text{joint}} / (p_0^{\text{joint}} + p_1^{\text{joint}})$; $p_1 \gets 1 - p_0$
        \If{$\textit{m} = \textit{sample}$}
          \State sample $b \in \{0,1\}$ with $\Pr[b = 0] = p_0$ and $\Pr[b = 1] = p_1$
        \ElsIf{$\textit{m} = \textit{preset}$}
          \State $b \gets$ next preset value for classical bit $c$
        \EndIf
        \State $p_{\text{global}} \gets p_{\text{global}} \times (b = 0\  ?\  p_0 : p_1)$
        \State $s \gets \textsc{MidMeasure}(s, q, b)$ \Comment{symbolic collapse on $q$}
        \State $\textit{clbits}[c] \gets b$
      \EndIf
    \EndProcedure
  \end{algorithmic}
\end{algorithm}

\noindent\emph{Interfaces used by the Simulator.}
Algorithm~\ref{alg:bddsim} relies on a small set of primitive operations
provided by the BDD-based Kernel and the \emph{Parser}:
\begin{itemize}[leftmargin=*]
  \item $\textsc{BDDCombSim}(n,r)$:
    create an empty BDD-based symbolic simulator for $n$ qubits with
    precision parameter $r$, i.e., the number of bit slices used to represent amplitudes in the encoding of~\cite{tsai2021bit}.
  \item $\textsc{IsFinalMeas}(G)$:
    return \emph{true} iff measurement $G$ is marked as \textit{final} by the
    OpenQASM~3 static analysis in the \emph{Parser}.
  \item $\textsc{GetProb}(s,q)$:
    for a BDD-encoded state $s$, return the unnormalised probabilities
    $(p^{\text{joint}}_0, p^{\text{joint}}_1)$ of outcomes $0$ and $1$
    on qubit $q$ (computed via weighted model counting in the \emph{Kernel}; Challenge~3).
  \item $\textsc{MidMeasure}(s,q,b)$:
    perform a symbolic mid-circuit measurement of qubit $q$ with outcome
    $b \in \{0,1\}$ on state $s$.
\end{itemize}

\paragraph{\textbf{Kernel (Boolean engine).}}
All quantum-state manipulation is delegated to the \emph{Kernel}, which is
implemented on top of BDDs.
The \emph{Kernel} provides the low-level symbolic operators used by the \emph{Simulator} and
is organised into two components, addressing Challenges~2 and~3.

\begin{itemize}[leftmargin=*]
  \item \texttt{BDDCombSim}, which implements Boolean gate-update rules for
        combinational circuits together with a mid-circuit measurement
        operator, following the BDD techniques of
        Tsai et al.~\cite{tsai2021bit}.  Combined with 
        CQC/DQC decomposition, this suffices to handle all combinational and
        dynamic fragments, as well as the mid-circuit and final measurements
        that arise in Qiskit programs.

  \item \texttt{BDDSeqSim}, which extends the \emph{Kernel} to full sequential quantum
        circuits in the sense of Wang et al.~\cite{wang2021equivalence}.
        Beyond Qiskit’s current \texttt{while\_loop} semantics, \texttt{BDDSeqSim}
        also supports an \emph{extended} \texttt{while}-loop style, where each
        iteration injects a fresh external input state into a fixed SQC body, e.g.:
\begin{lstlisting}[language=MyPython,basicstyle=\ttfamily\small,frame=single]
q_internal = ...    # internal qubits, persist across iterations
c = 0               # loop guard (classical bit)
while c == 0: 
    p = new_qubit()         # fresh external input with arbitrary state
    U(p, q_internal, ...)   # body unitary using p and internal qubits
    c = measure(p)          # update loop flag
\end{lstlisting}
        To simulate such general SQCs, \texttt{BDDSeqSim} introduces two
        dedicated operators: \emph{state composition}, which tensors the
        external input state with the current internal state, and
        \emph{state retention}, which discards the measured external qubits
        while preserving the updated internal state (Challenge~2).
        Their pseudo-code is given in Algorithm~\ref{alg:seqsim}.
\end{itemize}
\noindent\emph{Remark (internal vs.\ external qubits).}
For an SQC induced by a \texttt{while}-loop, \emph{internal} qubits are the
iteration-to-iteration state that is retained and fed back, whereas
\emph{external} qubits are iteration-scoped qubits that are measured and then
discarded.
\begin{algorithm}[H]  \caption{\textsc{BDDSeqSim}$(\ket{\psi_{\text{in}}}, \mathcal{C}, \mathcal{I}, \mathcal{M})$}
  \label{alg:seqsim}
  \textbf{Input:}
      Initial internal quantum state $\ket{\psi_{\text{in}}} = (k, \{F_x^i\})$ (notation as in~\cite{tsai2021bit});
      an SQC $\mathcal{C}$ whose body is a DQC $\mathcal{C}'$; a pair of sequences (or iterators) $\mathcal{I}$ and $\mathcal{M}$ providing, for each iteration, the external input state and the corresponding external measurement outcomes (given in \emph{preset} mode or generated on the fly in \emph{sample} mode).
      \\
    \textbf{Output:}
      Final internal state $(\hat{k}, \{\hat{F}_x^i\})$; path probability $p_{\text{global}}$ of $\mathcal{M}$.
  \begin{algorithmic}[1]
    
    \State $p_{\text{global}} \gets 1$ \Comment{accumulated probability of observing $\mathcal{M}$}
    \State $(k_{\text{in}}, \{F^{i}_{x,\text{in}}\}) \gets (k, \{F_x^i\})$ \Comment{initialise internal state}
    \For{each external input state $(0, \{F^{i}_{x,\text{ex}}\}) \in \mathcal{I}$}
      \State $(k_{\text{total}}, \{F^{i}_{x,\text{total}}\}) \gets (k_{\text{in}}, \{F^{0}_{d,\text{ex}} \land F^{i}_{x,\text{in}}\})$ 
        \Comment{state composition}
      \State apply $\mathcal{C}'$ to $(k_{\text{total}}, \{F^{i}_{x,\text{total}}\})$
        using Boolean gate transformations~\cite{tsai2021bit}
        and the \textsc{HandleMeasure} procedure of Algorithm~\ref{alg:bddsim} 
      \State $M[q_0],\ldots,M[q_{m-1}] \gets \mathcal{M}$\Comment{get the corresponding measurement outcomes}
      \State Construct assignment $\widetilde{q}_0\ldots\widetilde{q}_{m-1}$ based on $M[q_0],\ldots,M[q_{m-1}]$.
      \State $(k_\text{in}, \{F_{x,\text{in}}^i\})
    \gets \left(k_\text{total}, \left\{F_{x,\text{total}}^i \Big|_{\widetilde{q}_0 \dots \widetilde{q}_{m-1}}\right\}\right)$  \Comment{state retention}
      \State compute probability $p$ of observing $M[q_0],\ldots,M[q_{m-1}]$ 
             via weighted model counting (see \emph{Weighted model counting for probability evaluation})
      \State $p_{\text{global}} \gets p_{\text{global}}\cdot p$ 
    \EndFor
    \State $(\hat{k}, \{\hat{F}_x^i\}) \gets (k_{\text{in}}, \{F^{i}_{x,\text{in}}\})$
    \State \Return $\bigl[(\hat{k}, \{\hat{F}_x^i\}),\, p_{\text{global}}\bigr]$
  \end{algorithmic}
\end{algorithm}
\noindent\emph{Correctness.}
The state-composition and retention updates in Algorithm~\ref{alg:seqsim} realise, inside the Boolean amplitude encoding of Tsai et al.~\cite{tsai2021bit}, the tensor product of a computational-basis external input with the current internal state and the conditioning on the prescribed external measurement outcomes, respectively.
Appendix~\ref{app:seq-operators} gives an overview; Appendix~\ref{app:composition} proves the composition lemma, Appendix~\ref{app:mid} states and proves the mid-circuit measurement rule (with a multi-qubit corollary), Appendix~\ref{app:retention} formalises state retention (partial trace), and Appendix~\ref{app:rus-example} walks through one repeat-until-success iteration.

\noindent\emph{Remark (iterators for while loops).}
The \texttt{for}-loop in Algorithm~\ref{alg:seqsim} is conceptual: in the implementation,
the SQC is unfolded by repeatedly consuming one element from $\mathcal{M}$ (and, if needed, $\mathcal{I}$) \emph{while the loop guard holds}.
In \emph{sample} mode, $\mathcal{M}$ is an on-demand iterator whose next outcome is sampled when the guard is evaluated, and it stops producing values once the guard becomes false.
In \emph{preset} mode, the user provides a finite prefix of $\mathcal{M}$ (and $\mathcal{I}$), which fixes a bounded unrolling of the \texttt{while}-loop.

\paragraph{Weighted model counting for probability evaluation (Challenge~3).}
A central task of the \emph{Kernel} is to evaluate measurement probabilities on
BDD-encoded quantum states. Instead of using the hyper-function construction of \cite{tsai2021bit},
\emph{QSeqSim} adopts a model counting view: probabilities are obtained via
weighted model counting over structured and sparse BDDs representing the
amplitude polynomials (see~\cite{chavira2008probabilistic}).
Because BDDs compactly encode large families of basis states with shared
structure, weighted model counting can aggregate their contributions in time
proportional to the size of the BDD representation (up to caching and arithmetic
costs), avoiding explicit enumeration of exponentially many assignments. This
does not imply polynomial-time counting in the worst case with respect to the
original circuit; the gain comes from a compact BDD structure.
This probability computation is tightly integrated into
$\textsc{GetProb}$ and into the SQC engine (Algorithm~\ref{alg:seqsim}), and
forms the main algorithmic ingredient in addressing the scalability issues
highlighted in Challenge~3.
\section{Functionalities}
\label{sec:functionality}
We summarise here the main functionalities offered to users by \emph{QSeqSim}.
\subsection{Simulation of Qiskit While Loop and Its Extensions}
\emph{QSeqSim} can be used as a backend to simulate a broad class of Qiskit \texttt{while}-loops, including patterns that are not yet natively supported by current Qiskit.
\begin{itemize}
  \item \textbf{Qiskit \texttt{while\_loop} programs.}
  Users can directly run Qiskit circuits containing \texttt{while\_loop}
  constructs (possibly mixed with \texttt{if}/\texttt{else} and \texttt{for}).
  \emph{QSeqSim} executes the loop semantics and returns the resulting quantum
  state together with the probability of a chosen measurement pattern.

  \item \textbf{Extended while loops.}
  Beyond standard Qiskit \texttt{while}-loop programs, users can also specify SQCs that admit an arbitrary external input state at each iteration. \emph{QSeqSim} simulates such extended \texttt{while}-loops via \texttt{BDDSeqSim} and allows inspection of the state after any iteration.

  \item \textbf{General Qiskit programs with arbitrary control flow.}
  Qiskit programs with arbitrary combinations and nesting of \texttt{if}/\texttt{else},
  \texttt{switch}, \texttt{for}, and \texttt{while} constructs are supported.
\end{itemize}

In all cases, users may either let \emph{QSeqSim} randomly sample measurement
outcomes (\emph{sample} mode), or specify a concrete outcome pattern
(\emph{preset} mode) to fix a particular execution path.
\subsection{Analysis of Quantum State Reachability}

In formal methods, state reachability, namely deciding whether a given
configuration can be reached and with what probability, is a central problem,
from model checking of Markov chains and other probabilistic systems to the
verification of (quantum) while loops~\cite{akshay2015reachability,agrawal2015approximate,beauquier2002logic,guan2024measurement,gay2005probabilistic,feng2013model,ying2021model,chen2025autoq}.

\emph{QSeqSim} supports such (bounded) \emph{quantum state reachability} queries for
while loops. A program point is fixed
(for example, ``after the $k$-th loop iteration'') together with a target
quantum state or basis configuration on some qubits. For a quantum random
walk implemented with a \texttt{while}-loop, a reachability question of common
interest is:
\begin{quote}
``After the $k$-th iteration, is it possible for the walker to be at position
$i$? If so, what is the probability of the execution paths that realise this?''
\end{quote}
To answer such queries, candidate paths are encoded via the
\emph{preset} mode (fixing the sequence of measurement outcomes
across iterations). \emph{QSeqSim} then computes, for each path, the final
quantum state and its probability, which allows users to check
reachability and quantify the likelihood of reaching the target.

\subsection{Analysis of Measurement Outcome Reachability}
\emph{QSeqSim} also supports queries about whether particular \emph{measurement
outcome patterns} can occur and with what probability, a central concern in the analysis of quantum protocols with repeated
measurements and feedback, and in studying loop termination behaviour
\cite{eisert2012quantum,wang2021equivalence,ying2010quantum}.
Given a sequence of measurement outcomes across iterations for a while loop, specified via the
\emph{preset} measurement mode, \emph{QSeqSim} returns the probability of the
corresponding execution. For a RUS circuit expressed as a \texttt{while}-loop
that repeats until a measurement returns $0$, a typical measurement outcome
reachability question is:
\begin{quote}
``What is the probability that the sequence of outcomes is \texttt{1111110}
(that is, six times $1$ followed by $0$), and the loop terminates at that
point?''
\end{quote}
By varying such queries over families of patterns (for example, “eventually a
$0$ occurs”), one can empirically study whether a RUS or feedback‑based
protocol converges or may fail to terminate with non‑zero probability.

\section{Evaluation}
\label{sec:evaluation}
We evaluate \emph{QSeqSim} along three dimensions matching its intended uses:
(i) simulation of Qiskit \texttt{while}-loops (Section~\ref{subsec:sim-while}),
(ii) quantum state reachability (Section~\ref{subsec:state-reach}),
and (iii) measurement outcome reachability (Section~\ref{subsec:meas-reach}).
All experiments were run on a MacBook Air 2023 (Apple M2, 16 GB unified memory)
with Python~3.12.
\subsection{Simulation of Qiskit While Loops}
\label{subsec:sim-while}
We consider three structured benchmarks (RUS~\cite{paetznick2013repeat}, QRW~\cite{shenvi2003quantum}, Grover~\cite{andres2022weakly,grover1996fast}) and, separately, randomly
generated \texttt{while}-loop programs.
Unless stated otherwise, we use \emph{preset} mode: a concrete sequence of
mid-circuit measurement outcomes is fixed in advance, so that
\emph{QSeqSim} symbolically executes a single predetermined path.

For the RUS and random \texttt{while}-loop benchmarks, circuits are generated
via the Qiskit Python API, exported to OpenQASM~3, and then passed to
\emph{QSeqSim} as in Section~3. In contrast, QRW and Grover rely in Qiskit on
custom Python functions that repeatedly instantiate large static subcircuits
(shift, oracle, diffusion), incurring significant OpenQASM parsing and
construction overhead. To isolate the cost of symbolic simulation, we therefore
bypass Qiskit for QRW and Grover and directly invoke the \texttt{BDDSeqSim}
kernel on hand-specified sequential circuits, using the same BDD-based engine.
\paragraph{Repeat-until-success (RUS).}
We use the six two-qubit RUS circuits from Paetznick and Svore~\cite[Figures~7--10]{paetznick2013repeat}, where a common RUS pattern is instantiated with different
single-qubit unitaries $U$ and corresponding ancilla measurement rules. We implement these six RUS variants (\textsf{RUS-1}
to \textsf{RUS-6}) in the Qiskit API and invoke \emph{QSeqSim} in \emph{preset} mode. In each case the ancilla is measured every iteration and
the loop terminates on a designated “success’’ outcome: for \textsf{RUS-1}--\textsf{RUS-4}
success is 0, while for \textsf{RUS-5} and \textsf{RUS-6} success
is 1. Accordingly, we fix the ancilla outcome pattern to
$0^{k-1}1$ for \textsf{RUS-1}--\textsf{RUS-4} and to $1^{k-1}0$ for
\textsf{RUS-5}/\textsf{RUS-6}, enforcing exactly $k$ iterations.

Table~\ref{tab:sim-while-rus} reports a representative configuration for six
different RUS implementations $U$.
The number of qubits and per-iteration gate count remain tiny (two qubits,
$7$--$37$ gates), so runtimes mainly reflect the symbolic cost of unfolding the
loop. For short loops (10 iterations), all variants complete in under $1$s.
For 100 iterations, however, runtimes span nearly three orders of magnitude:
from below 1.5s for \textsf{RUS-1}/\textsf{RUS-2} to around 90--108s for \textsf{RUS-5}/\textsf{RUS-6}.
This variability shows that \emph{QSeqSim} is sensitive not only to loop depth
but also to the specific unitary $U$, which affects BDD structure and sharing.
\begin{table}[htp]
\centering
\caption{Performance of \emph{QSeqSim} on Qiskit RUS \texttt{while}-loops (\emph{preset} mode).}
\label{tab:sim-while-rus}
\begin{tabular}{|c|c|c|c|c|c|}
\hline
Benchmark & Implementation $U$ & \#Qubits & \#Gates & \#Iterations & Time (s) \\
\hline
\multirow{2}{*}{\textsf{RUS-1}}
  & \multirow{2}{*}{\scalebox{0.8}{$\displaystyle\frac{2X+\sqrt{2}Y+Z}{\sqrt{7}}$}}
  & \multirow{2}{*}{2}
  & \multirow{2}{*}{7} & 10   & 0.300    \\
\cline{5-6}
  &  &  &  & 100  & 0.960   \\
\hline
\multirow{2}{*}{\textsf{RUS-2}}
  & \multirow{2}{*}{\scalebox{0.8}{$\displaystyle\frac{I+2i\sqrt{2}X}{\sqrt{3}}$}}
  & \multirow{2}{*}{2}
  & \multirow{2}{*}{13} & 10   & 0.673    \\
\cline{5-6}
  &  &  &  & 100  & 1.427   \\
\hline
\multirow{2}{*}{\textsf{RUS-3}}
  & \multirow{2}{*}{\scalebox{0.8}{$\displaystyle\frac{I+2iZ}{\sqrt{5}}$}}
  & \multirow{2}{*}{2}
  & \multirow{2}{*}{12} & 10   & 0.656    \\
\cline{5-6}
  &  &  &  & 100  & 11.081   \\
\hline
\multirow{2}{*}{\textsf{RUS-4}}
  & \multirow{2}{*}{\scalebox{0.8}{$\displaystyle\frac{3I+2iZ}{\sqrt{13}}$}}
  & \multirow{2}{*}{2}
  & \multirow{2}{*}{20} & 10   & 0.760    \\
\cline{5-6}
  &  &  &  & 100  & 20.807   \\
\hline
\multirow{2}{*}{\textsf{RUS-5}}
  & \multirow{2}{*}{\scalebox{0.8}{$\displaystyle\frac{4I+iZ}{\sqrt{17}}$}}
  & \multirow{2}{*}{2}
  & \multirow{2}{*}{37} & 10   & 0.963    \\
\cline{5-6}
  &  &  &  & 100  & 90.027   \\
\hline
\multirow{2}{*}{\textsf{RUS-6}}
  & \multirow{2}{*}{\scalebox{0.8}{$\displaystyle\frac{5I+2iZ}{\sqrt{29}}$}}
  & \multirow{2}{*}{2}
  & \multirow{2}{*}{33} & 10   & 0.994    \\
\cline{5-6}
  &  &  &  & 100  & 108.148  \\
\hline
\end{tabular}
\end{table}
\paragraph{Quantum random walk (QRW).}
For QRW, we work directly with \texttt{BDDSeqSim}, following the experimental
setup of Wang et al.~\cite[Figure~3]{wang2021equivalence}. Each iteration applies
a Hadamard on the coin, a coin-controlled shift on an $\ell$-qubit position
register, and a multi-controlled $X$ on a single flag qubit that flips when a
fixed target basis state is reached, namely the internal all-ones configuration
of the coin and the position register.

Table~\ref{tab:sim-while-qrw} shows runtimes as we scale
the position register from $N=16$ to $N=1024$ sites and increase the number of
iterations. For $N=16$ (16 qubits, 31 gates per iteration), \emph{QSeqSim}
comfortably reaches hundreds of iterations, but the cost grows rapidly:
going from 100 to 785 iterations increases the time from 5.3s to about 1800s.
Larger walks are limited by qubit and gate counts rather than loop depth:
for $N=1024$ (1024 qubits, 2047 gates per iteration), only up to 13 iterations
are feasible within a 30-minute timeout.

\paragraph{Grover's search.}
The Grover benchmark follows the same pattern: a fixed number of data qubits
plus a flag qubit, with each iteration consisting of an oracle phase flip on a
single marked basis state followed by a diffusion operator.
Again we emulate a \texttt{while}-loop by enforcing a specific termination
iteration via the mid-circuit measurement pattern.

Table~\ref{tab:sim-while-grover} summarises the results.
For 16 work qubits (81 gates per iteration), runtimes scale from $0.18$s at
3 iterations to around 1780s near 200 iterations.
For 128 qubits, even 15 iterations already saturate our timeout.
At 256 qubits and 1281 gates per iteration, \emph{QSeqSim} handles up to 5
iterations (about 19 minutes) before timeout (30 minutes) at 6 iterations.
These trends are consistent with QRW: BDD-based symbolic simulation can handle
substantial gate counts and deep loops, but very large multi-qubit operators
combine unfavourably with repeated iteration.
\begin{table}[htp]
    \centering
    \caption{Performance of \emph{QSeqSim} on different \texttt{while}-loops (\texttt{BDDSeqSim}).}
    \label{tab:sim-while-combined}
    \begin{minipage}[t]{0.48\textwidth}
        \centering
        \subcaption{QRW \texttt{while}-loops}
        \label{tab:sim-while-qrw}
        \begin{tabular}{|c|c|c|c|}
            \hline
            \#Qubits & \#Gates & \#Iterations & Time (s) \\ 
            \hline
            \multirow{5}{*}{16}  & \multirow{5}{*}{31}   & 3   & 0.041 \\
            \cline{3-4}
                                 &                       & 10  & 0.111 \\
            \cline{3-4}
                                 &                       & 100 & 5.303 \\
            \cline{3-4}
                                 &                       & 200 & 30.375 \\ 
            \cline{3-4}
                                 &                       & 785 & 1797.676 \\ 
            \hline
            \multirow{5}{*}{128} & \multirow{5}{*}{255}  & 3   & 0.698 \\
            \cline{3-4}
                                 &                       & 10  & 2.861 \\
            \cline{3-4}
                                 &                       & 100 & 239.485 \\
            \cline{3-4}
                                 &                       & 200 & 1118.501 \\ 
            \cline{3-4}
                                 &                       & 244 & 1792.934 \\ 
            \hline
            \multirow{5}{*}{1024}& \multirow{5}{*}{2047} & 3   & 235.568 \\
            \cline{3-4}
                                 &                       & 5   & 438.156 \\
            \cline{3-4}
                                 &                       & 8   & 867.763 \\
            \cline{3-4}
                                 &                       & 10  & 1124.445 \\ 
            \cline{3-4}
                                 &                       & 13  & 1622.304 \\ 
            \hline
        \end{tabular}
    \end{minipage}
    \hfill 
    \begin{minipage}[t]{0.48\textwidth}
        \centering
        \subcaption{Grover \texttt{while}-loops}
        \label{tab:sim-while-grover}
        \begin{tabular}{|c|c|c|c|}
            \hline
            \#Qubits & \#Gates & \#Iterations & Time (s) \\
            \hline
            \multirow{5}{*}{16}  & \multirow{5}{*}{81}   & 3   & 0.177 \\
            \cline{3-4}
                                 &                       & 10  & 2.354 \\
            \cline{3-4}
                                 &                       & 100 & 340.669 \\
            \cline{3-4}
                                 &                       & 150 & 866.182 \\ 
            \cline{3-4}
                                 &                       & 203 & 1781.141 \\ 
            \hline
            \multirow{5}{*}{128} & \multirow{5}{*}{641}  & 2   & 3.115 \\
            \cline{3-4}
                                 &                       & 3   & 16.293 \\
            \cline{3-4}
                                 &                       & 6   & 135.108 \\
            \cline{3-4}
                                 &                       & 10  & 622.699 \\ 
            \cline{3-4}
                                 &                       & 15  & 1700.021 \\ 
            \hline
            \multirow{5}{*}{256} & \multirow{5}{*}{1281} & 2   & 256.593 \\
            \cline{3-4}
                                 &                       & 3   & 527.699 \\ 
            \cline{3-4}
                                 &                       & 4   & 840.786 \\ 
            \cline{3-4}
                                 &                       & 5   & 1156.052 \\ 
            \cline{3-4}
                                 &                       & 6   & Timeout \\ 
            \hline
        \end{tabular}
    \end{minipage}
\end{table}

\paragraph{Random quantum while loops.}
To evaluate robustness under unstructured control flow, we generated random quantum circuits using the Qiskit API, which contain a \texttt{while}-loop with additional quantum gates placed outside the loop. We fixed the system size at 100 qubits and varied two key experimental parameters: computational density (defined as the total number of gates in the circuit, denoted $G$) and control complexity (defined as the total number of mid-circuit measurements, denoted $M$). These experiments were executed in \emph{sample} mode: for each mid-circuit measurement performed within the loop, \emph{QSeqSim} computes the exact outcome distribution over the current symbolic state, then samples the next execution branch, and proceeds until the loop terminates or reaches a predefined iteration bound of 1000.

As in \emph{sample} mode the loop termination time is a random variable determined by the sampled measurement outcomes, the resulting runtime distribution is typically right-skewed; we therefore repeat each configuration 50 times and report the median and interquartile range (IQR $[Q_1,Q_3]$)\footnote{The median is the 0.5 quantile of the sample distribution, and the interquartile range $[Q_1,Q_3]$ is the interval between the 0.25 and 0.75 quantiles; see~\cite{hoaglin2000understanding}.}. Table~\ref{tab:sim-while-random} shows that \emph{QSeqSim} remains tractable on
100-qubit circuits: across all
tested settings, median runtimes stay within a minute. Runtime generally grows with
circuit size, while mid-circuit measurements have a non-monotone, control-flow-dependent
effect (e.g., $G{=}200$, $M{=}10$ vs.\ $20$). As random \texttt{while}-loops may
not terminate, our generator enforces almost termination; these results thus
characterize robustness under terminating control flow.

\begin{table}[htp]
\centering
\caption{Performance of \emph{QSeqSim} on Qiskit random \texttt{while}-loops (\emph{sample} mode).}
\label{tab:sim-while-random}
  \begin{tabular}{|l|c|c|c|c|}
    \hline
    Benchmark & Qubits & Gates & Mid-Meas & Total Time (s) \\ \hline
    \textsf{rqc\_q100\_g50\_m5}   & 100 & 50  & 5  & 0.218 \scriptsize{[0.215, 0.915]} \\ \hline
    \textsf{rqc\_q100\_g50\_m10}  & 100 & 50  & 10 & 0.273 \scriptsize{[0.194, 1.158]} \\ \hline
    \textsf{rqc\_q100\_g50\_m20}  & 100 & 50  & 20 & 0.214 \scriptsize{[0.211, 0.367]} \\ \hline
    \textsf{rqc\_q100\_g100\_m5}  & 100 & 100 & 5  & 0.971 \scriptsize{[0.945, 3.053]} \\ \hline
    \textsf{rqc\_q100\_g100\_m10} & 100 & 100 & 10 & 0.949 \scriptsize{[0.942, 2.789]} \\ \hline
    \textsf{rqc\_q100\_g100\_m20} & 100 & 100 & 20 & 1.149 \scriptsize{[1.008, 6.092]} \\ \hline
    \textsf{rqc\_q100\_g200\_m5}  & 100 & 200 & 5  & 2.235 \scriptsize{[2.151, 5.910]} \\ \hline
    \textsf{rqc\_q100\_g200\_m10} & 100 & 200 & 10 & 15.978 \scriptsize{[3.056, 22.853]} \\ \hline
    \textsf{rqc\_q100\_g200\_m20} & 100 & 200 & 20 & 2.116 \scriptsize{[1.767, 10.228]} \\ \hline
\end{tabular}
\end{table}
\subsection{Analysis of Quantum State Reachability}
\label{subsec:state-reach}

We next use the QRW benchmark to study quantum state reachability for a
\texttt{while}-loop semantics. Each QRW step consists of a coin flip on a
single qubit, a controlled shift on an $\ell$-qubit position register, and a
multi-controlled $X$ that flips an external flag qubit $q_0$ precisely when the
coin and position jointly equal a fixed target basis state (here, the
all-ones internal configuration).
\paragraph{Reachability formulation.}
A reachability question of common
interest is:
\begin{quote}
\emph{``After $k$ QRW iterations, can the internal configuration (coin + position)
reach the all-ones basis state, and if so, with what exact probability?''}
\end{quote}
This reduces to the marginal state of the
flag qubit $q_0$ after $k$ steps: the all-ones state is reachable at iteration
$k$ iff $q_0$ can be in $\ket{1}$.
Under \emph{preset} mode, we fix the measurement outcomes controlling the loop guard to
$0^{k-1}1$, so that the loop continues for $k-1$ iterations and terminates at
step $k$. The branch probability of this classical pattern is then exactly the
reachability probability of the all-ones internal configuration at step $k$.
\paragraph{Results.}
Table~\ref{tab:state-reach-qrw} reports results for a fixed system size of
256 qubits and iteration counts $k = 1,\dots,10$.
For each $k$, \texttt{BDDSeqSim} symbolically executes the constrained path
$0^{k-1}1$, checks whether measuring qubit $q_0$ in the computational basis
can yield outcome $1$, and returns the exact branch probability together
with the runtime.

\begin{table}[htp]
  \centering
  \caption{Quantum state reachability of QRW \texttt{while}-loop.}
  \label{tab:state-reach-qrw}
  \begin{tabular}{|c|c|c|c|c|}
    \hline
    \#Qubits & $k$ & Reachable? & Probability & Time (s) \\
    \hline
    \multirow{10}{*}{256}
      & $1$  & Yes & 0.5 & 1.264 \\ \cline{2-5}
      & $2$  & No  & 0 & 2.489 \\ \cline{2-5}
      & $3$  & Yes & 0.125 & 3.808 \\ \cline{2-5}
      & $4$  & No  & 0 & 5.053 \\ \cline{2-5}
      & $5$  & No  & 0 & 6.399 \\ \cline{2-5}
      & $6$  & No  & 0 & 7.992 \\ \cline{2-5}
      & $7$  & Yes & 0.0078125 & 9.729 \\ \cline{2-5}
      & $8$  & No  & 0 & 11.409 \\ \cline{2-5}
      & $9$  & No  & 0 & 13.453 \\ \cline{2-5}
      & $10$ & No  & 0 & 15.563 \\
    \hline
  \end{tabular}
\end{table}

The non-zero rows reveal that the all-ones configuration is only reachable at
specific iteration counts (here, $k = 1, 3, 7$), and the associated branch
probabilities decay quickly.
\emph{QSeqSim} computes these probabilities exactly, without sampling noise,
and the runtimes grow approximately linearly with $k$ in this setting.
\subsection{Analysis of Measurement Outcome Reachability}
\label{subsec:meas-reach}

We finally evaluate how easily users can pose measurement-outcome reachability
queries on the \textsf{RUS-1} \texttt{while}-loop from
Section~\ref{subsec:sim-while}. Each loop iteration applies the fixed
two-qubit \textsf{RUS-1} subcircuit and measures an ancilla; the loop
continues while the ancilla is $1$ and terminates on the first $0$.

A typical measurement outcome reachability query is:
\begin{quote}
  \emph{``What is the probability that the ancilla outcomes are $1^k 0$, so that the loop terminates
  exactly at iteration $k+1$?''}
\end{quote}

In \emph{preset} mode, the user specifies the concrete pattern $1^k 0$;
\emph{QSeqSim} then symbolically follows this unique branch and returns:
(i) the probability of that pattern, $p_\text{path}(k)$, and
(ii) the probability that the loop has terminated by iteration $k+1$,
$p_\text{term.}(k)$.
Here $p_\text{term.}(k)$ is obtained as the cumulative probability of all
termination patterns up to iteration $k+1$, i.e.\
$p_\text{term.}(k) = \sum_{i=0}^{k} p_\text{path}(i)$.

\begin{table}[htp]
  \centering
  \caption{Measurement-outcome reachability for the \textsf{RUS-1} \texttt{while}-loop.}
  \label{tab:rus-meas-reach}
  \begin{tabular}{|c|c|c|c|c|c|c|c|c|c|}
    \hline
    $k$             & 0       & 1       & 2       & 3       & 4       & 5       & 6       & 7       & 8       \\
    \hline
    $p_\text{path}$ & 0.75000 & 0.18750 & 0.04688 & 0.01172 & 0.00293 & 0.00073 & 0.00018 & 0.00005 & 0.00001 \\
    \hline
    $p_\text{term.}$& 0.75000 & 0.93750 & 0.98438 & 0.99609 & 0.99902 & 0.99976 & 0.99994 & 0.99998 & 1.00000 \\
    \hline
    Time (s)        & 0.02607 & 0.05073 & 0.04982 & 0.04946 & 0.07638 & 0.11201 & 0.15824 & 0.21086 & 0.29526 \\
    \hline
  \end{tabular}
\end{table}

Table~\ref{tab:rus-meas-reach} shows the outputs for $k = 0,\dots,8$.
The per-pattern probabilities $p_\text{path}$ quickly become very small as $k$
increases, while termination is already almost certain after a few iterations.
All queries complete in well under one second, allowing interactive exploration
of rare measurement patterns.

\section{Conclusion}
\label{sec:conclusion}

We have presented \emph{QSeqSim}, a Qiskit-integrated symbolic backend that directly executes \texttt{while}-loop programs by interpreting them as sequential quantum circuits with explicit state feedback across iterations. Building on a BDD-based Boolean representation of quantum states and operations, \emph{QSeqSim} provides symbolic operators tailored to this sequential setting and thus enables efficient exploration of loop-induced behaviours. Our experiments on RUS schemes, quantum random walks, and Grover's search on commodity hardware show that \emph{QSeqSim} scales to substantial \texttt{while}-induced sequential circuits; in particular, in the quantum random walk benchmark we successfully simulate circuits with over 1000 qubits for more than 10 loop iterations, and we support both quantum state reachability and measurement outcome reachability queries for Qiskit \texttt{while}-loop programs. All experiments in this work are conducted on a PC. As future work, we plan to deploy \emph{QSeqSim} on server-class machines and exploit parallelism to support larger-scale simulations and better serve subsequent formal verification of quantum programs with \texttt{while}-loops.

\begin{credits}
\subsubsection{\ackname} We sincerely thank the anonymous reviewers for their valuable comments and constructive suggestions. This work was partially supported by the Innovation Program for Quantum Science and Technology (Grant No. 2024ZD0300500), the Youth Innovation Promotion Association, Chinese Academy of Sciences (Grant No. 2023116), the National Natural Science Foundation of China (Grant No. 62402485), the Young Elite Scientists Sponsorship Program, China Association for Science and Technology, the CCF-QuantumCtek Superconducting Quantum Computing Special Cooperation Program (Grant No. CCF-QC2025007), and the International Partnership Program of the Chinese Academy of Sciences (Grant No. 096GJHZ2025013FN).
\end{credits}

\bibliographystyle{splncs04}
\bibliography{ref}

\appendix

\begingroup
\setlength{\parskip}{0.55\baselineskip plus 0.12\baselineskip minus 0.08\baselineskip}

\section{Overview}
\label{app:seq-operators}

This appendix records lemmas on the $(k,\{F_x^i\})$ encoding that justify the \emph{state composition} and \emph{state retention} steps in Algorithm~\ref{alg:seqsim}, together with the mid-circuit measurement rule used inside the dynamic body $\mathcal{C}'$ and in \textsc{HandleMeasure} (Algorithm~\ref{alg:bddsim}).

\paragraph{Quantum registers.}
Fix $n=m+l$ for one SQC iteration and partition the $n$ qubit lines:
\[
  Q[n]=Q_{\mathrm{ex}}[m]\cup Q_{\mathrm{in}}[l],
  \qquad
  Q_{\mathrm{ex}}[m]\cap Q_{\mathrm{in}}[l]=\emptyset,
\]
where $Q_{\mathrm{ex}}[m]=\{q_0,\ldots,q_{m-1}\}$ and $Q_{\mathrm{in}}[l]=\{q_m,\ldots,q_{n-1}\}$ are the external (input and measured) and internal (feedback) registers.
Set $\mathcal{H}_{\mathrm{ex}}=\bigotimes_{j=0}^{m-1}\mathcal{H}_{q_j}$ and $\mathcal{H}_{\mathrm{in}}=\bigotimes_{j=m}^{n-1}\mathcal{H}_{q_j}$; the joint Hilbert space is $\mathcal{H}_{\mathrm{ex}}\otimes\mathcal{H}_{\mathrm{in}}$.

\begin{figure}[H]
\centering
\includegraphics[width=0.85\linewidth]{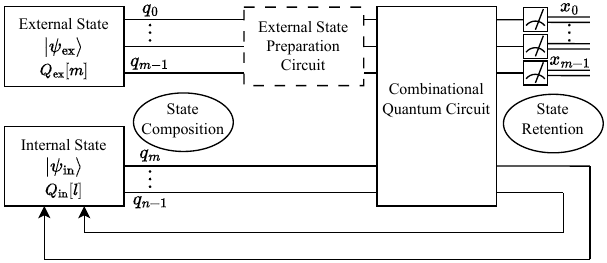}
\caption{Circuit-level view of one SQC iteration (composition, dynamic body with measurement on the external register, retention).}
\label{fig:app-sqc}
\end{figure}

All statements are phrased in the $(k,\{F_x^i\})$ encoding of Tsai et al.~\cite{tsai2021bit} and use the register partition above together with the external/internal naming of Section~\ref{sec:tool-architecture}.
We focus on the lemmas that interface directly with \emph{QSeqSim}; longer worked examples (for instance, stabiliser-style error syndromes with repeated mid-circuit measurements) follow the same disjunctive-normal-form bookkeeping used below.

\paragraph{Soundness of one SQC iteration (sketch).}
In Algorithm~\ref{alg:seqsim}, the composition line matches Theorem~\ref{thm:app-tensor}; the body $\mathcal{C}'$ is simulated with the same Boolean gate and mid-circuit updates as Algorithm~\ref{alg:bddsim}, for which Tsai et al.~\cite{tsai2021bit} and Theorem~\ref{thm:app-mid} supply the local correctness of each primitive; the retention line matches Theorem~\ref{thm:app-trace} once external outcomes are fixed.
Chaining these updates for every iteration therefore preserves the intended sequential semantics in the $(k,\{F_x^i\})$ encoding.

\bigskip

\section{State composition}
\label{app:composition}

In each SQC iteration the external register is prepared in a computational basis state (possibly after a preparatory CQC on external qubits only), while the internal register may carry an arbitrary entangled state from the previous iteration.

\begin{theorem}
\label{thm:app-tensor}
Let $\ket{\psi_{\mathrm{ex}}}=(0,\{F_{x,\mathrm{ex}}^i\})$ be a computational basis state on an $m$-qubit external register and let $\ket{\psi_{\mathrm{in}}}=(k_{\mathrm{in}},\{F_{x,\mathrm{in}}^i\})$ be an arbitrary state on a disjoint $l$-qubit internal register.
Then
\[
  \ket{\psi_{\mathrm{ex}}}\otimes \ket{\psi_{\mathrm{in}}}
  = \bigl(k_{\mathrm{in}},\,\{F_{d,\mathrm{ex}}^0 \land F_{x,\mathrm{in}}^i\}\bigr).
\]
\end{theorem}

\begin{proof}
Let $m$-qubit state $\ket{\psi_{\mathrm{ex}}} = (0, \{F_{x,\mathrm{ex}}^i\})$ be the computational basis state $\ket{j}$, then the amplitude is given by $\alpha^{\mathrm{ex}}_s = \delta_{sj}$.
Here, $\delta_{sj}$ is the Kronecker delta.
For $l$-qubit state $\ket{\psi_{\mathrm{in}}}$, the amplitude is given by $\alpha^{\mathrm{in}}_t = \frac{1}{\sqrt{2}^{k_{\mathrm{in}}}}(a_{t,\mathrm{in}}\omega^3 + b_{t,\mathrm{in}}\omega^2 + c_{t,\mathrm{in}}\omega + d_{t,\mathrm{in}})$ in the normal form of Tsai et al.~\cite{tsai2021bit}.

Thus, defining $A=(\alpha^{\mathrm{ex}}_s)_{0\leq s <2^m}^T,B=(\alpha^{\mathrm{in}}_t)_{0\leq t<2^l}^T$, the amplitude of their tensor product $A\otimes B=(\beta_{s\ll l+t})_{s,t}^T$ is given by:
\begin{align*}
\beta_{s\ll l+t} &= \alpha^{\mathrm{ex}}_s \cdot \alpha^{\mathrm{in}}_t \\
&= \delta_{sj} \cdot \frac{1}{\sqrt{2}^{k_{\mathrm{in}}}}(a_{t,\mathrm{in}}\omega^3 + b_{t,\mathrm{in}}\omega^2 + c_{t,\mathrm{in}}\omega + d_{t,\mathrm{in}}) \\
&= \frac{1}{\sqrt{2}^{k_{\mathrm{in}}}}(\delta_{sj}a_{t,\mathrm{in}}\omega^3 + \delta_{sj}b_{t,\mathrm{in}}\omega^2 + \delta_{sj}c_{t,\mathrm{in}}\omega + \delta_{sj}d_{t,\mathrm{in}}).
\end{align*}

Thus, the scalar of the tensor product state $\ket{\psi_{\mathrm{ex}}} \otimes \ket{\psi_{\mathrm{in}}}$ is $k_{\mathrm{in}}$, and the components of the amplitude $\beta_{s\ll l+t}$ are given by $x_{s \ll l + t} = \delta_{sj}\, x_{t,\mathrm{in}}$.
Since $\delta_{sj}$ is either $0$ or $1$, for each bit position $i$, we have:
\[
x_{s \ll l + t}(i) = \delta_{sj} \land x_{t,\mathrm{in}}(i).
\]
This corresponds to the Boolean function:
\[
F_x^i = F_{d,\mathrm{ex}}^0 \land F_{x,\mathrm{in}}^i,
\]
where $F_{d,\mathrm{ex}}^0$ encodes the binary vector $(\delta_{0j}, \ldots, \delta_{2^m-1,j})^T$ and $F_{x,\mathrm{in}}^i$ encodes the binary vector $(x_{0,\mathrm{in}}(i), \ldots, x_{2^l-1,\mathrm{in}}(i))^T$.

Therefore, the tensor product state is represented by $(k_{\mathrm{in}}, \{F_{d,\mathrm{ex}}^0 \land F_{x,\mathrm{in}}^i\})$, which completes the proof of the theorem.
\qed
\end{proof}

\paragraph{Remark.}
Theorem~\ref{thm:app-tensor} is stated for external computational basis states $\ket{\psi_{\mathrm{ex}}}$, matching the composition line of Algorithm~\ref{alg:seqsim}.
More generally, if the intended external input is $\ket{\psi_{\mathrm{ex}}}=U\ket{e}$ for a preparatory CQC $U$ on $Q_{\mathrm{ex}}[m]$ and a computational basis state $\ket{e}$, then
\[
  \ket{\psi_{\mathrm{ex}}}\otimes\ket{\psi_{\mathrm{in}}}
  \;=\;
  (U\otimes I)\,\bigl(\ket{e}\otimes\ket{\psi_{\mathrm{in}}}\bigr),
\]
so the same semantics can be recovered symbolically by first forming $\ket{e}\otimes\ket{\psi_{\mathrm{in}}}$ via Theorem~\ref{thm:app-tensor} and then simulating $U\otimes I$ on the joint tuple with the usual Boolean gate updates~\cite{tsai2021bit}; equivalently, external initialization may be deferred until after the composition step whenever that is convenient for the symbolic tensor construction (cf.\ the dashed box in Figure~\ref{fig:app-sqc}).

\bigskip

\section{Mid-circuit measurement}
\label{app:mid}

The following single-qubit statement is the Boolean analogue of projecting onto a fixed computational outcome; it underpins \textsc{MidMeasure} in Algorithm~\ref{alg:bddsim}.

\begin{theorem}
\label{thm:app-mid}
When a quantum measurement $M$ is performed on qubit $q_t$ of a state $\ket{\psi}=(k, \{F_x^i\})$, the resulting post-measurement state is $(k, \{\hat{F_x^i}\})$, where
\[
\hat{F_x^i} = \widetilde{q}_t F_x^i \big|_{\widetilde{q}_t}, \text{ if } \widetilde{q}_t = \begin{cases}
    \overline{q_t}, & \text{if } M[q_t] = 0, \\
    q_t, & \text{if } M[q_t] = 1.
\end{cases}
\]
\end{theorem}

\begin{proof}
We begin by expressing $F_x^i$ in its fully expanded disjunctive normal form, so that every term explicitly contains all Boolean variables. This expansion enables us to partition the terms of $F_x^i$ into two distinct classes:
\begin{itemize}
    \item Terms containing $\overline{\widetilde{q}_t}$: Any term that contains the literal $\overline{\widetilde{q}_t}$ will evaluate to \texttt{false} once the operation $F_x^i \big|_{\widetilde{q}_t}$ is applied. This is because, in accordance with the predetermined measurement outcome, the literal $\overline{\widetilde{q}_t}$ is replaced by its negated value, thereby invalidating the entire term.
    \item Terms containing $\widetilde{q}_t$: For every term that includes $\widetilde{q}_t$, the application of the substitution $\widetilde{q}_tF_x^i\big|_{\widetilde{q}_t}$ leaves the term unchanged. This invariance occurs because the measured qubit $q_t$ retains its original Boolean value consistent with the measurement outcome.
\end{itemize}

Thus, after applying the measurement operation across all terms, the resultant expression $\hat{F_x^i}$ retains only those terms from $F_x^i$ that include $\widetilde{q}_t$. These surviving terms are exactly those that are consistent with the outcome of the measurement $M[q_t]$. Therefore, $\hat{F_x^i}$ is equivalent to the projection of the quantum state onto the subspace corresponding to the predetermined measurement outcome. This completes the proof.
\qed
\end{proof}

\begin{corollary}
\label{cor:app-mid-multi}
When a quantum measurement $M$ is performed on quantum register $Q_{\mathrm{ex}}[m]=\{q_0,\ldots,q_{m-1}\}$ of a quantum state $\ket{\psi}=(k, \{F_x^i\})$, the resulting post-measurement state becomes $(k, \{\hat{F_x^i}\})$, where
\[
\hat{F_x^i} = \widetilde{q}_0 \ldots \widetilde{q}_{m-1} F_x^i \big|_{\widetilde{q}_0 \ldots \widetilde{q}_{m-1}}, \text{ where for } 0\leq j< m, \widetilde{q}_j = \begin{cases}
    \overline{q_j}, & \text{if } M[q_j] = 0, \\
    q_j, & \text{if } M[q_j] = 1.
\end{cases}
\]
\end{corollary}

\begin{proof}
This follows by applying Theorem~\ref{thm:app-mid} successively to $q_0,\ldots,q_{m-1}$.
\qed
\end{proof}

\bigskip

\section{State retention}
\label{app:retention}

\begin{theorem}
\label{thm:app-trace}
Let $\ket{\psi} = (k, \{F_x^i\})$ be a state on quantum register $Q_{\mathrm{ex}}[m]\cup Q_{\mathrm{in}}[l]$ with $Q_{\mathrm{ex}}[m]\cap Q_{\mathrm{in}}[l]=\emptyset$.
When $\ket{\psi}$ is measured by a measurement $M$ on $Q_{\mathrm{ex}}[m] = \{q_0, \dots, q_{m-1}\}$, the resulting state on $Q_{\mathrm{in}}[l]$ is $\ket{\psi_{\mathrm{in}}} = (k, \{\hat{F_x^i}\})$, where
\[
\hat{F_x^i} = F_x^i \Big|_{\widetilde{q}_0 \dots \widetilde{q}_{m-1}},\text{ where for } 0\leq j< m,\; \widetilde{q}_j = \begin{cases}
    \overline{q_j}, & \text{if } M[q_j] = 0, \\
    q_j, & \text{if } M[q_j] = 1.
\end{cases}
\]
\end{theorem}

\begin{proof}
We begin by expressing $F_x^i$ in its fully expanded disjunctive normal form, so that every term explicitly contains all Boolean variables. This expansion enables us to partition the terms of $F_x^i$ into two distinct classes:
\begin{itemize}
    \item Terms containing $\widetilde{q}_0\dots\widetilde{q}_{m-1}$: For every term that includes the sequence $\widetilde{q}_0\dots\widetilde{q}_{m-1}$, the application of the substitution $F_x^i \big|_{\widetilde{q}_0\dots\widetilde{q}_{m-1}}$
    leaves the unmeasured qubits unchanged, thereby preserving their original Boolean values.
    
    \item All other terms: Any term that does not include the full sequence $\widetilde{q}_0\dots\widetilde{q}_{m-1}$ necessarily contains an occurrence of $\overline{\widetilde{q}_t}$ for some $0 \leq t < m$. When the operation $F_x^i \big|_{\widetilde{q}_0\dots\widetilde{q}_{m-1}}$ is applied, such terms evaluate to \texttt{false}, in accordance with the predetermined measurement outcome.
\end{itemize}

Thus, after applying the operation to all terms, the resultant expression $\hat{F_x^i}$ retains only those terms from $F_x^i$ that contain $\widetilde{q}_0\dots\widetilde{q}_{m-1}$. These surviving terms correspond exactly to the components of the quantum state that are consistent with the predefined measurement outcome, while the contributions from the measured qubits are eliminated. Therefore, $\hat{F_x^i}$ is equivalent to the partial trace of the quantum state over the measured qubits, projecting the state onto the subspace defined by the measurement outcome. This completes the proof.
\qed
\end{proof}

\paragraph{Relation to mid-circuit measurement.}
Corollary~\ref{cor:app-mid-multi} describes the \emph{in-body} Boolean update after a computational-basis measurement on $Q_{\mathrm{ex}}[m]$ has been resolved to fixed outcomes $M[q_0],\ldots,M[q_{m-1}]$.
Quantum mechanically, this is the usual selective (Lüders) update on the joint Hilbert space $\mathcal{H}_{\mathrm{ex}}\otimes\mathcal{H}_{\mathrm{in}}$: the state is conditioned onto the joint eigenspace fixed by those $M[q_j]$, while the literals $\widetilde{q}_j$ supplied in Corollary~\ref{cor:app-mid-multi} mark that branch in the $(k,\{F_x^i\})$ encoding \emph{without} tracing out $\mathcal{H}_{\mathrm{ex}}$, because coherent gates in~$\mathcal{C}'$ may still couple $Q_{\mathrm{ex}}[m]$ to $Q_{\mathrm{in}}[l]$ before the iteration ends.
Theorem~\ref{thm:app-trace}, by contrast, targets the iteration boundary once the same $M[q_j]$ are fixed: only the internal subsystem is carried to the next round, i.e.\ the reduced density operator obtained by $\mathrm{Tr}_{\mathcal{H}_{\mathrm{ex}}}$ from the joint post-measurement state.
Recording $F_x^i|_{\widetilde{q}_0\cdots\widetilde{q}_{m-1}}$ and dropping the explicit product $\widetilde{q}_0\cdots\widetilde{q}_{m-1}$ at retention therefore matches updating the interface to the marginal internal state passed into the following tensor-product step, rather than offering an alternative reading of Corollary~\ref{cor:app-mid-multi} as a purely syntactic Boolean rewrite inside~$\mathcal{C}'$.

\bigskip

\section{Worked example: one RUS iteration}
\label{app:rus-example}

We next give an example that illustrates Boolean state composition and retention in the sense of Theorems~\ref{thm:app-tensor} and~\ref{thm:app-trace}.
The SQC in Figure~\ref{fig:app-rus} realises an $X$ gate from the Repeat-Until-Success construction~\cite{paetznick2013repeat}: when the measurement on $q_0$ returns outcome~$1$, an $X$ rotation is applied to~$q_1$, and the procedure is repeated until that success event occurs.

\begin{figure}[H]
\centering
\includegraphics[width=0.35\textwidth]{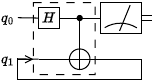}
\caption{A SQC for the Repeat-Until-Success protocol of implementing $X$ gate.}
\label{fig:app-rus}
\end{figure}

In Figure~\ref{fig:app-rus}, the internal state $\ket{\psi_{\mathrm{in}}}$ on $q_1$ is initialised to $\ket{0}$, encoded as $\ket{\psi_{\mathrm{in}}}=(0,\{F_{x,\mathrm{in}}^i\})$, while all other Boolean functions are \texttt{false}.
Here
\[
  F_{d,\mathrm{in}}^0=\overline{q_1}.
\]

The external input on $q_0$ is $\ket{\psi_{\mathrm{ex}}}=\ket{0}$, encoded as $\ket{\psi_{\mathrm{ex}}}=(0,\{F_{x,\mathrm{ex}}^i\})$, while all other Boolean functions are also \texttt{false}, with
\[
  F_{d,\mathrm{ex}}^0=\overline{q_0}.
\]

By Theorem~\ref{thm:app-tensor}, composition yields the total state on $\{q_0,q_1\}$,
\[
  \ket{\psi_{\mathrm{ex}}}\otimes\ket{\psi_{\mathrm{in}}}
  =\bigl(0,\{F_{d,\mathrm{ex}}^0\land F_{x,\mathrm{in}}^i\}\bigr)
  =\bigl(0,\{F_{x,\text{total}}^i\}\bigr),
\]
with all other Boolean functions remaining \texttt{false}, and
\[
  F_{d,\text{total}}^0=\overline{q_0q_1}.
\]

Subsequently, the total state undergoes Boolean updates through quantum gates $H$ and $\mathrm{CNOT}_{q_0\rightarrow q_1}$ following the gate rules in~\cite{tsai2021bit}, yielding $(1,\{F_{x,\text{total}}^i\})$ while all other Boolean functions remain \texttt{false}, with
\[
  F_{d,\text{total}}^0=q_0q_1\lor \overline{q_0q_1}.
\]

If the predefined measurement outcome is $M[q_0]=1$, Theorem~\ref{thm:app-trace} updates the tuple through retention to a new internal state $(1,\{\hat{F}_{x,\mathrm{in}}^i\})$ on $q_1$, while the other Boolean functions remain \texttt{false}, and
\[
  \hat{F}_{d,\mathrm{in}}^0=F_{d,\text{total}}^0\big|_{q_0}=q_1.
\]
Thus one iteration concludes with $(1,\{\hat{F}_{x,\mathrm{in}}^i\})$ as the internal input for the next round.

\endgroup

\end{document}